\documentclass[10pt]{article}
\usepackage{setspace}

\usepackage[a4paper]{geometry}
\usepackage{silence}
\usepackage{tikz}
\usetikzlibrary{arrows.meta, positioning, calc}

\usepackage[english]{babel}
\usepackage[utf8]{inputenc}
\usepackage[T1]{fontenc}

\usepackage[dvipsnames]{xcolor}

\usepackage{amsmath, amssymb, amsfonts, amsthm, mathtools}
\usepackage{mathrsfs}
\usepackage{braket}
\usepackage{cancel}
\usepackage{latexsym}

\usepackage{graphicx}
\usepackage{float}
\usepackage{subcaption}
\usepackage{listings}
\usepackage{stackengine}

\usepackage{blindtext}
\usepackage{verbatim}
\usepackage{array}
\usepackage{setspace}
\usepackage{multicol}
\usepackage{fancybox}
\usepackage{rotating}
\usepackage{enumerate}
\usepackage{textcomp}
\usepackage{sectsty}
\usepackage{comment}
\usepackage{pifont}
\usepackage{empheq}
\usepackage{titlesec}
\usepackage{abstract}
\usepackage[autostyle=true]{csquotes}

\newtheorem{theorem}{Theorem}
\newtheorem{proposition}{Proposition}

\usepackage{tikz}
\usetikzlibrary{arrows.meta,shapes}

\definecolor{micolor}{RGB}{230,240,255}
\xdefinecolor{ritsumeikan}{RGB}{217, 180, 34}
\definecolor{darkbackground}{RGB}{0, 52, 64}
\xdefinecolor{colormio}{RGB}{217, 180, 34}

\tikzset{
  mynode/.style={
    draw,
    chamfered rectangle,
    minimum width=3.5cm,
    minimum height=0.7cm,
    align=center,
    fill=micolor!50
  },
  myarrow/.style={
    -{Stealth[length=3mm]},
    thick,
    shorten >=8pt,
    shorten <=8pt
  },
  myarrowdbl/.style={
    {Stealth[length=3mm]}-{Stealth[length=3mm]},
    thick,
    shorten >=8pt,
    shorten <=8pt
  },
  pathlabel/.style={midway, sloped, fill=white, inner sep=1pt, font=\footnotesize},
  pathnum/.style={midway, sloped, circle, draw, fill=white, inner sep=1pt, font=\scriptsize}
  route/.style={draw, thick, color=blue},
}

\usepackage[most]{tcolorbox}
\tcbset{
  my math box/.style={
    standard jigsaw,
    colback=ritsumeikan!10, colframe=ritsumeikan, size=fbox, arc=3pt, boxrule=0.5pt,
    opacityback=0.82,
  },
  my math boxy/.style={
    standard jigsaw,
    colback=ritsumeikan!20, colframe=ritsumeikan, size=fbox, arc=3pt, boxrule=0.5pt,
    opacityback=0.82,
  }
}

\newtcbox{\mymath}[1][]{%
    nobeforeafter, math upper, tcbox raise base,
    enhanced, colframe=blue!30!black,
    colback=blue!30, boxrule=1pt,
    #1}

\titleformat{\paragraph}{\normalfont\normalsize\bfseries}{}{0pt}{}
\titlespacing*{\paragraph}{0pt}{1em}{0.7em}

\makeatletter
\renewcommand\@makefnmark{\hbox{\@textsuperscript{\normalfont\color{colormio!75!black}\@thefnmark}}}
\makeatother

\usepackage{etoolbox}
\makeatletter
\pretocmd{\tagform@}{\color{gray!50!black}}{}{}
\makeatother

\usepackage[backend=biber,style=numeric-comp,sorting=none,sortcites]{biblatex}
\usepackage{hyperref}
\hypersetup{
  colorlinks = true,
  linkcolor  = blue!80!black, 
  citecolor  = red,       
  filecolor  = red,
  urlcolor   = black
}

\usepackage[normalem]{ulem}

\usepackage{setspace}
\begin{document}
\title{\Large{\textbf{Reconstruction Between Generalized Hybrid Metric--Palatini Gravity and $\Phi(R,\phi,X)$ Theories}}}

\author{   \small{Jonathan Ramírez\footnote{\href{mailto:ramirez.jsg@gmail.com}{ramirez.jsg@gmail.com}}}\\
 \small{
 \textit{Departamento de Física Teórica, Instituto de Física,}}\\
\textit{\small{Universidade do Estado do Rio de Janeiro,}}\\
\textit{\small{Rua São Francisco Xavier 524, Maracanã,}} \\ 
\textit{\small{CEP 20550-013, Rio de Janeiro, Brazil.}}
}

\date{}
\maketitle

\begin{abstract}
We develop a local reconstruction framework between \(\Phi(R,\phi,X)\) theories with linear dependence on \(X\) and generalized hybrid metric--Palatini gravity. The construction is formulated in vacuum in the Einstein frame, where both formulations can be written as two-scalar theories with the same field-space geometry. The framework provides a practical method for finding \(\Phi(R,\phi,X)\) and \(f(R,\mathcal R)\) functions that describe the same regular Einstein-frame two-scalar sector. Starting from a given \(\Phi(R,\phi,X)\) model, we derive the equation that determines the compatible hybrid functions \(f(R,\mathcal R)\) and show that it has a Clairaut-type structure. We also show that the inverse reconstruction is not unique: a regular hybrid Einstein-frame potential determines a family of compatible \(\Phi(R,\phi,X)\) theories, parametrized by the kinetic coupling. Explicit examples illustrate the reconstruction procedure, its domain of validity, and the translation of model parameters between the \(\Phi(R,\phi,X)\) and \(f(R,\mathcal R)\) formulations.
\end{abstract}

\section{Introduction}
\label{sec:Introduction}

Modified theories of gravity provide a systematic framework for testing the geometrical foundations of General Relativity (GR) and for exploring possible
extensions of the Einstein--Hilbert action in regimes where additional gravitational degrees of freedom may become relevant \cite{Clifton2011,Nojiri2010,Berti2015}. Among the best studied examples are metric \(f(R)\) theories, in which the gravitational Lagrangian is promoted to a nonlinear function of the Ricci scalar \cite{Sotiriou:2008rp,felice}. Such models also provide a useful starting point for more general scalar--tensor
extensions, where extra scalar modes may arise either from nonlinear curvature terms or from independent scalar fields already present in the action.

A broad generalization is obtained by allowing the Lagrangian to depend on the metric Ricci scalar \(R\), on a scalar field \(\phi\), and on its kinetic
density \(X\), leading to theories of the form \(\Phi(R,\phi,X)\). This class includes generalized \(f(R,\phi)\) models and k-essence-type constructions
\cite{Bahamonde2015,Nojiri2019,Nilok2009,Sotiriou:2011dz}, and has been used in cosmological applications ranging from early- and late-time acceleration
\cite{Amendola1993,Shinji2003} to perturbation theory
\cite{Hwang2002,Garriga:1999vw} and exact gravitational configurations \cite{Bahamonde2018,Malik2022}. In the present work we focus on the subclass in which the dependence on \(X\) is linear. Under suitable local invertibility and positivity conditions, this subclass admits an Einstein-frame representation as GR coupled to two interacting scalar fields
\cite{Maeda1989,Sho2015,Dhimiter2020,Anirudh2020}.

A similar two-scalar Einstein-frame structure arises in generalized hybrid
metric--Palatini (GH) gravity. In this theory the action depends on both the
metric Ricci scalar \(R\) and a Palatini scalar \(\mathcal R\), the latter
being constructed from an independent torsionless connection \cite{tamanini}.
Its scalar--tensor representation can be formulated as a bi-scalar--tensor
theory which, in the Einstein frame, describes GR coupled to two scalar degrees
of freedom. The resulting scalar field-space geometry coincides with that of
the \(\Phi(R,\phi,X)\) subclass considered here.

This common Einstein-frame kinetic structure therefore motivates a reconstruction problem between the two formulations. However, sharing the same kinetic geometry in the Einstein frame does not by itself imply an equivalence between the original Jordan-frame actions. In both theories, the scalar--tensor description is valid only on regular branches. In the \(\Phi(R,\phi,X)\) case, one must require, in particular, local invertibility of the curvature sector, a positive conformal factor, and positivity of the
kinetic coupling. In the GH case, one must require positivity of the Einstein-frame conformal factor, together
with the non-degeneracy of the Legendre map associated with \(f(R,\mathcal R)\).

Consequently, the relevant question is not whether the two Jordan-frame actions are globally equivalent, but whether their regular scalar--tensor branches can be reconstructed locally so as to generate the same Einstein-frame two-field sector. This reconstruction is relevant for modified-gravity model building because it shows how the same regular Einstein-frame two-field dynamics can be represented by different Jordan-frame actions, namely by a \(\Phi(R,\phi,X)\) theory or by a generalized hybrid metric--Palatini theory. The point of the present construction is not merely to compare two Einstein-frame actions with the same kinetic geometry, but to formulate the comparison at the level of regular scalar--tensor branches, where the distinct Legendre maps, conformal transformations, and positivity conditions of the two Jordan-frame theories must be treated explicitly. In this way, the reconstruction identifies the conditions under which the correspondence is admissible.

The aim of this work is to formulate this reconstruction problem constructively in vacuum. Starting from a regular \(\Phi(R,\phi,X)\) theory
with linear dependence on \(X\), we derive the Clairaut-type equation that
determines the corresponding generalized hybrid functions \(f(R,\mathcal R)\), when they exist. We show that the complete linear integral of this equation is degenerate from the GH scalar--tensor point of view, since its Hessian determinant vanishes. Hence, admissible GH branches must be obtained from the singular envelope and must still satisfy the required positivity and non-degeneracy conditions. This criterion separates formal solutions of the reconstruction equation from genuinely regular hybrid scalar--tensor branches.

We also study the inverse problem. Given a regular GH Einstein-frame potential, we reconstruct the compatible \(\Phi(R,\phi,X)\) representatives. The inverse map is not unique: it depends on the choice of scalar-field coordinate, or equivalently on kinetic coupling. Thus, a single regular GH Einstein-frame scalar sector determines an equivalence class of \(\Phi(R,\phi,X)\) theories rather than a unique Jordan-frame action. This
non-uniqueness is a structural feature of the correspondence and reflects the freedom in parametrizing the independent scalar field in the
\(\Phi(R,\phi,X)\) formulation.

The present construction is technically connected to previous reconstruction approaches for Einstein-frame scalar sectors in modified gravity and, in
particular, to the relation between higher-derivative scalar--tensor theories and GH gravity discussed in \cite{Ramirez2026}. However, the class of theories considered here is structurally different. In the previous case, the additional
scalar degrees of freedom are generated by higher-derivative curvature terms in the gravitational sector. In the present work, instead, the starting point is a \(\Phi(R,\phi,X)\) theory, where \(\phi\) is an independent scalar field and \(X\) is its kinetic density. The reconstruction problem is therefore formulated for a different Jordan-frame functional dependence.

This paper is organized as follows. In Sec.~\ref{sec:EF-Representation of GST} we review the Einstein-frame representation of the relevant \(\Phi(R,\phi,X)\) theories. In Sec.~\ref{sec:GH-gravity} we summarize the
scalar--tensor formulation of GH gravity. In Sec.~\ref{sec:Reconstruction-Framework} we prove the forward and inverse reconstruction results and analyze the Clairaut branches. In Sec.~\ref{sec:Applications} we present the explicit applications. Finally, in Sec.~\ref{sec:Conclusions}, we summarize our main results.

\section{\texorpdfstring{Einstein-Frame Representation of $\Phi(R,\phi,X)$ Gravity}~}
\label{sec:EF-Representation of GST}

We consider the vacuum sector of a class of modified gravity theories described by the action\footnote{With $\kappa^2 = 8\pi G$, $c=1$.}
\begin{equation}
    S=\frac{1}{2\kappa^2}\int d^4x\sqrt{-g}\,
    \Phi\left(R,\phi,X\right),
    \label{action-f(R,phi,X)}
\end{equation}
where 
\begin{equation}
    X\equiv-\frac{1}{2}g^{\mu\nu}\nabla_{\mu}\phi\nabla_{\nu}\phi .
\end{equation}
In this work we restrict attention to the subclass in which the dependence on the
kinetic term is linear,
\begin{equation}
    \Phi(R,\phi,X)=h(R,\phi)+P(\phi)X-U(\phi),
    \label{f(R,phi,X)}
\end{equation}
with \(h(R,\phi)\) nonlinear in \(R\). The functions \(P(\phi)\) and \(U(\phi)\)
are otherwise arbitrary. This subclass admits a scalar--tensor representation
which, under the regularity conditions specified below, can be expressed in the
Einstein frame as a two-scalar theory \cite{Maeda1989,Anirudh2020}. The linear
dependence on \(X\) is essential for obtaining, after an appropriate field
redefinition, the two-scalar kinetic structure employed in the reconstruction
procedure developed in this work.

To make the scalar--tensor structure explicit, we introduce an auxiliary field
\(\xi\) and write a dynamically equivalent Jordan-frame action as
\begin{equation}
S_J=\frac{1}{2\kappa^2}\int d^4x\,\sqrt{-g}\left[
h(\xi,\phi)+\gamma(R-\xi)
-\frac{1}{2}P(\phi)\,g^{\mu\nu}\nabla_{\mu}\phi\nabla_{\nu}\phi
-U(\phi)\right],
\label{action-Jframe}
\end{equation}
where \(\gamma\) is a Lagrange multiplier. Variation with respect to \(\gamma\)
enforces \(\xi=R\), while variation with respect to \(\xi\) yields
\begin{equation}
    \gamma=h_{,\xi}(\xi,\phi),
    \label{Theta}
\end{equation}
where the comma in the subscript denotes a partial derivative with respect to the indicated argument, i.e., $h_{,\xi}\equiv \partial h/\partial \xi$.
We assume that this relation can be locally inverted, namely
\begin{equation}
    h_{,\xi\xi}(\xi,\phi)\neq0 .
    \label{invertibility-h}
\end{equation}
On this regular branch, one can solve for \(\xi=\xi(\gamma,\phi)\), and the action
takes the scalar--tensor form
\begin{equation}
S_J=\int d^4x\,\sqrt{-g}\left[
\frac{\gamma}{2\kappa^2}R
-\frac{1}{4\kappa^2}P(\phi)\,
g^{\mu\nu}\nabla_{\mu}\phi\nabla_{\nu}\phi
-U_J(\gamma,\phi)
\right],
\label{eq:Jordan-ST-Phi}
\end{equation}
where
\begin{equation}
U_J(\gamma,\phi)=\frac{1}{2\kappa^2}
\left[
\gamma\,\xi(\gamma,\phi)
-h(\xi(\gamma,\phi),\phi)
+U(\phi)
\right].
\label{potential-UJ}
\end{equation}

The passage to the Einstein frame requires a positive conformal factor. We therefore
restrict to the branch
\begin{equation}
    \gamma=h_{,\xi}(\xi,\phi)>0.
    \label{condition-gamma-positive}
\end{equation}
Applying the conformal transformation
\begin{equation}
    \tilde g_{\mu\nu}=\gamma\,g_{\mu\nu},
\end{equation}
the action becomes
\begin{equation}
S_E=\int d^4x\,\sqrt{-\tilde g}
\left[
\frac{\tilde R}{2\kappa^2}
-\frac{3}{4\kappa^2\gamma^2}
\tilde g^{\mu\nu}\tilde\nabla_{\mu}\gamma\tilde\nabla_{\nu}\gamma
-\frac{1}{4\kappa^2}\gamma^{-1}P(\phi)
\tilde g^{\mu\nu}\tilde\nabla_{\mu}\phi\tilde\nabla_{\nu}\phi
-\gamma^{-2}U_J(\gamma,\phi)
\right].
\label{eq:Einstein-P}
\end{equation}
Introducing the field
\begin{equation}
    \gamma=e^{\sqrt{\frac{2}{3}}\kappa\chi},
    \label{PHI}
\end{equation}
the kinetic term associated with the conformal mode becomes canonical.

Finally, we introduce a new scalar degree of freedom \(\sigma\) through an
invertible field redefinition,
\begin{equation}
\sigma=\mathcal J(\phi)
\equiv
\frac{1}{\sqrt{2}\,\kappa}
\int^\phi \sqrt{P(\bar\phi)}\,d\bar\phi+C.
\label{sigma-P(phi)}
\end{equation}
This transformation is real and locally invertible provided
\begin{equation}
    P(\phi)>0,
    \qquad
    \mathcal J'(\phi)\neq0 .
    \label{condition-P-positive}
\end{equation}
The integration constant \(C\) has no physical significance in the general
construction and may be absorbed into a shift of \(\sigma\). In explicit examples
we set \(C=0\) unless otherwise specified.

With these definitions, the action assumes the two-field Einstein-frame form
\begin{equation}
S_E=
\int d^4x\sqrt{-\tilde g}
\left[
\frac{\tilde R}{2\kappa^2}
-\frac{1}{2}\tilde g^{\mu\nu}
\nabla_\mu\chi\nabla_\nu\chi
-\frac{1}{2}e^{-\sqrt{\frac{2}{3}}\kappa\chi}
\tilde g^{\mu\nu}\nabla_\mu\sigma\nabla_\nu\sigma
-\tilde W(\chi,\sigma)
\right],
\label{eq:action-biscalar}
\end{equation}
where the Einstein-frame potential is given by
\begin{equation}
\tilde W(\chi,\sigma)
=\frac{1}{2\kappa^2}e^{-2\sqrt{\frac{2}{3}}\kappa\chi}
\left[e^{\sqrt{\frac{2}{3}}\kappa\chi}\xi\!\left(
e^{\sqrt{\frac{2}{3}}\kappa\chi},\mathcal J^{-1}(\sigma)
\right)-h\!\left(\xi\!\left(e^{\sqrt{\frac{2}{3}}\kappa\chi},\mathcal J^{-1}(\sigma)\right),
\mathcal J^{-1}(\sigma)\right)+U\!\left(\mathcal J^{-1}(\sigma)\right)\right].
\label{potential-f-r-phi-X}
\end{equation}

Equivalently, the scalar fields parametrize a two-dimensional field space with
line element
\begin{equation}
    d\ell^2=d\chi^2+
    e^{-\sqrt{\frac{2}{3}}\kappa\chi}d\sigma^2 .
    \label{eq:field-space-Phi}
\end{equation}
In the Einstein-frame coordinates \((\chi,\sigma)\), the field-space metric
takes the universal form \eqref{eq:field-space-Phi}, while the model-dependent
information is entirely encoded in the potential \(\tilde W(\chi,\sigma)\).
This structural property enables a direct comparison between the
Einstein-frame scalar sector of \(\Phi(R,\phi,X)\) gravity and the corresponding
scalar--tensor sector of GH gravity.

Accordingly, the regular Einstein-frame branch of the theory
\eqref{f(R,phi,X)} is characterized by the conditions
\begin{equation}
    h_{,R}>0,
    \qquad
    h_{,RR}\neq0,
    \qquad
    P(\phi)>0 .
    \label{eq:regularity-Phi-summary}
\end{equation}
These conditions define the domain of validity of the bidirectional
reconstruction framework formulated in Section \ref{sec:Reconstruction-Framework}.

\section{Generalized Hybrid Metric--Palatini Gravity}
\label{sec:GH-gravity}

GH gravity is defined by an action depending on both the
metric Ricci scalar \(R\) and the Palatini scalar \(\mathcal R\),
\begin{equation}
    S=\frac{1}{2\kappa^2}\int d^4x\sqrt{-g}\,
    f(R,\mathcal R),
    \label{action-GH}
\end{equation}
where \(R=g^{\mu\nu}R_{\mu\nu}\) is constructed from the Levi--Civita connection
of \(g_{\mu\nu}\), while \(\mathcal R=g^{\mu\nu}\mathcal R_{\mu\nu}\) is built
from an independent torsionless connection \(\widetilde\Gamma^\rho_{\mu\nu}\).
The corresponding Ricci tensor is
\begin{equation}
\mathcal R_{\mu\nu}
=
\widetilde{\Gamma}^{\lambda}_{~\mu\nu,\lambda}
-\widetilde{\Gamma}^{\lambda}_{~\mu\lambda,\nu}
+\widetilde{\Gamma}^{\sigma}_{~\mu\nu}
 \widetilde{\Gamma}^{\lambda}_{~\sigma\lambda}
-\widetilde{\Gamma}^{\sigma}_{~\mu\lambda}
 \widetilde{\Gamma}^{\lambda}_{~\sigma\nu}.
\end{equation}
In this section we recall the scalar--tensor representation of the theory,
following \cite{tamanini}, and identify the regular branch relevant for the
reconstruction developed below.

Variation of Eq.~\eqref{action-GH} with respect to the metric yields
\begin{equation}
    f_{,R}R_{\mu\nu}
    -\frac{1}{2}g_{\mu\nu}f
    -\left(\nabla_{\mu}\nabla_{\nu}-g_{\mu\nu}\Box\right)f_{,R}
    +f_{,\mathcal R}\mathcal R_{(\mu\nu)}
    =0
    \label{field-equation}
\end{equation}
Variation with respect to the independent connection gives
\begin{equation}
    \widetilde{\nabla}_{\rho}
    \left(
    \sqrt{-g}\,f_{,\mathcal R}g^{\mu\nu}
    \right)=0.
    \label{connection-equation}
\end{equation}
where $\widetilde{\nabla}_{\mu}$ denotes the covariant derivative in terms of $\widetilde{\Gamma}^\rho_{\mu\nu}$. Eq.~\eqref{connection-equation} 
can be solved by introducing $\sqrt{-h}h^{\mu\nu}=\sqrt{-g}f_{,\mathcal R}g^{\mu\nu}$, leading to
\begin{equation}
    \widetilde{\Gamma}_{\nu \lambda}^\mu=\frac{1}{2}h^{\mu\rho} (\partial_\nu h_{\rho \lambda}+\partial_\lambda h_{\rho \nu}-\partial_\rho h_{\nu \lambda} ) \label{gamma-palatini}.
\end{equation}
Hence, the connection 
$\widetilde{\Gamma}_{\nu \lambda}^\mu$ can be written as a Levi--Civita connection in terms of the conformal metric 
$h_{\mu\nu}=f_{,\mathcal R}g_{\mu\nu}$. Using Eq.~\eqref{gamma-palatini}, the Ricci tensor associated
with the connection $\tilde\Gamma$ can be written as
\begin{equation}
\mathcal R_{\mu\nu}=R_{\mu\nu}+\frac{3}{2(f_{,\mathcal R})^2}\partial_\mu f_{,\mathcal R} \partial_\nu f_{,\mathcal R}-\frac{1}{ f_{,\mathcal R}}\left(\nabla_{\mu}\nabla_{\nu}+\frac{1}{2}g_{\mu\nu}\Box\right)f_{,\mathcal R}.
\label{ricci-calig}
\end{equation}

To obtain the scalar--tensor representation of the GH theory in the absence of matter fields, two auxiliary scalar fields $\alpha$ and $\beta$ are introduced such that the action in Eq.~\eqref{action-GH} reads \cite{tamanini}
\begin{equation}
S=
\frac{1}{2\kappa^2}
\int d^4x\sqrt{-g}
\left[
f(\alpha,\beta)
+f_{,\alpha}(R-\alpha)
+f_{,\beta}(\mathcal R-\beta)
\right].
\label{action-ab}
\end{equation}
Variation with respect to \(\alpha\) and \(\beta\) yields
\begin{equation}
\begin{aligned}
f_{,\alpha\alpha}(R-\alpha)
+
f_{,\alpha\beta}(\mathcal R-\beta)&=0,\\
f_{,\alpha\beta}(R-\alpha)
+
f_{,\beta\beta}(\mathcal R-\beta)&=0.
\end{aligned}
\label{determinant-equations}
\end{equation}
Hence, the identification
\[
\alpha=R,\qquad \beta=\mathcal R
\]
is locally enforced provided the Hessian determinant is nonvanishing,
\begin{equation}
\Delta_f
\equiv
f_{,\alpha\alpha}f_{,\beta\beta}
-
f_{,\alpha\beta}^{\,2}
\neq0 .
\label{constrain-hessiano}
\end{equation}
This condition corresponds to the non-degeneracy of the Legendre map. It
excludes, in particular, functions linear in \(R\) and \(\mathcal R\), and plays
a central role in the classification of the Clairaut branches discussed below.

We now define the scalar fields
\begin{equation}
    \varphi=f_{,\alpha},
    \qquad
    \psi=-f_{,\beta}.
    \label{field-hybrid}
\end{equation}
The minus sign in the definition of \(\psi\) ensures that the scalar field
associated with the Palatini sector acquires the standard sign in the
Einstein-frame kinetic term on the branch \(\psi>0\). The scalar potential is
defined by
\begin{equation}
    V(\varphi,\psi)
    =
    -f(\alpha(\varphi,\psi),\beta(\varphi,\psi))
    +\varphi\,\alpha(\varphi,\psi)
    -\psi\,\beta(\varphi,\psi).
    \label{potencial-V-hibrida}
\end{equation}
The action then takes the form
\begin{equation}
S=
\frac{1}{2\kappa^2}
\int d^4x\sqrt{-g}
\left[
\varphi R-\psi\mathcal R
-
V(\varphi,\psi)
\right].
\label{action-ST-GH-before-connection}
\end{equation}

Taking the trace of Eq.~\eqref{ricci-calig}, it follows that
\begin{equation}
\mathcal R=R+\frac{3}{2\psi^2}\nabla_\mu\psi\nabla^\mu\psi
-\frac{3}{\psi}\Box\psi .
\label{Palatini-scalar-psi}
\end{equation}
Substituting Eq.~\eqref{Palatini-scalar-psi} into
Eq.~\eqref{action-ST-GH-before-connection} and discarding a total divergence,
one obtains the Jordan-frame scalar--tensor representation
\begin{equation}
S=
\frac{1}{2\kappa^2}
\int d^4x\sqrt{-g}
\left[
(\varphi-\psi)R
-\frac{3}{2\psi}\nabla_\mu\psi\nabla^\mu\psi
-
V(\varphi,\psi)
\right].
\label{action-ST-GH}
\end{equation}
Defining
\begin{equation}
    \vartheta=\varphi-\psi,
    \qquad
    W(\vartheta,\psi)=V(\vartheta+\psi,\psi),
\end{equation}
the action becomes
\begin{equation}
S=
\frac{1}{2\kappa^2}
\int d^4x\sqrt{-g}
\left[
\vartheta R
-\frac{3}{2\psi}\nabla_\mu\psi\nabla^\mu\psi
-
W(\vartheta,\psi)
\right].
\label{action-ST-GH-vartheta}
\end{equation}
The passage to the Einstein frame requires
\begin{equation}
    \vartheta>0,
    \qquad
    \psi>0.
    \label{regularity-GH-Jordan}
\end{equation}
The first condition ensures a positive conformal factor, while the second
guarantees that the field redefinition associated with the Palatini scalar is
real.

Applying the conformal transformation
\begin{equation}
    \tilde g_{\mu\nu}=\vartheta g_{\mu\nu}
\end{equation}
and introducing the fields
\begin{equation}
    \chi=
    \sqrt{\frac{3}{2\kappa^2}}\ln\vartheta,
    \qquad
    \sigma=
    \frac{\sqrt6}{\kappa}\sqrt{\psi},
    \label{campos-TH}
\end{equation}
the action assumes the Einstein-frame form
\begin{equation}
S=
\int d^4x\sqrt{-\tilde g}
\left[
\frac{\tilde R}{2\kappa^2}
-\frac{1}{2}\tilde g^{\mu\nu}
\nabla_\mu\chi\nabla_\nu\chi
-\frac{1}{2}
e^{-\sqrt{\frac{2}{3}}\kappa\chi}
\tilde g^{\mu\nu}\nabla_\mu\sigma\nabla_\nu\sigma
-\tilde W_f(\chi,\sigma)
\right],
\label{Rosa-action}
\end{equation}
where
\begin{equation}
\tilde W_f(\chi,\sigma)
=
\frac{1}{2\kappa^2}
e^{-2\sqrt{\frac{2}{3}}\kappa\chi}
V\left(
e^{\sqrt{\frac{2}{3}}\kappa\chi}
+\frac{\kappa^2\sigma^2}{6},
\frac{\kappa^2\sigma^2}{6}
\right).
\label{potencial-hibirda.scalar-tensor}
\end{equation}

The form of Eq.~\eqref{Rosa-action} is the point at which the comparison with the \(\Phi(R,\phi,X)\) sector becomes direct. On the regular branch, the Einstein-frame scalar sector of GH gravity takes the same kinetic form as that obtained for the
\(\Phi(R,\phi,X)\) theories discussed in
Sec.~\ref{sec:EF-Representation of GST}. In the variables \((\chi,\sigma)\),
the corresponding scalar field-space line element is
\begin{equation}
    d\ell^2  =d\chi^2+e^{-\sqrt{\frac{2}{3}}\kappa\chi}d\sigma^2.
    \label{eq:field-space-GH}
\end{equation}
Thus, the distinction between the two formulations is not in the kinetic sector
but in the Einstein-frame potential, which is \(\tilde W_f(\chi,\sigma)\) for
the GH theory. This observation provides the basis for the
reconstruction maps developed in the next section.

Combining Eqs.~\eqref{field-hybrid} and \eqref{regularity-GH-Jordan}, the
regularity domain of the hybrid representation can be expressed directly in
terms of \(f(R,\mathcal R)\) as
\begin{equation}
    f_{,R}+f_{,\mathcal R}>0,
    \qquad
    f_{,\mathcal R}<0,
    \qquad
    f_{,RR}f_{,\mathcal R\mathcal R}
    -
    f_{,R\mathcal R}^{\,2}
    \neq0 .
    \label{eq:regularity-GH-summary}
\end{equation}
These conditions will be used below to distinguish the regular (singular-envelope)
Clairaut branch from the degenerate linear branch.

\section{Reconstruction Framework}
\label{sec:Reconstruction-Framework}

The preceding sections show that the class of \(\Phi(R,\phi,X)\) theories
linear in \(X\) and GH gravity lead, on their respective regular branches, to
Einstein-frame scalar sectors with the same functional form of the kinetic
term. In both cases the scalar fields span a two-dimensional field space with
line element
\begin{equation}
    d\ell^2=d\chi^2+
    e^{-\sqrt{\frac{2}{3}}\kappa\chi}d\sigma^2 .
\end{equation}
Consequently, once the regular branches have been specified, the model-dependent information relevant for the correspondence is encoded in the Einstein-frame scalar potential. The resulting branch-level reconstruction scheme is summarized in Fig.~\ref{fig:correspondence}. This section formulates
the matching as two local reconstruction theorems, one in each direction.

\begin{figure}[h]
    \centering
    \includegraphics[width=0.8\linewidth]{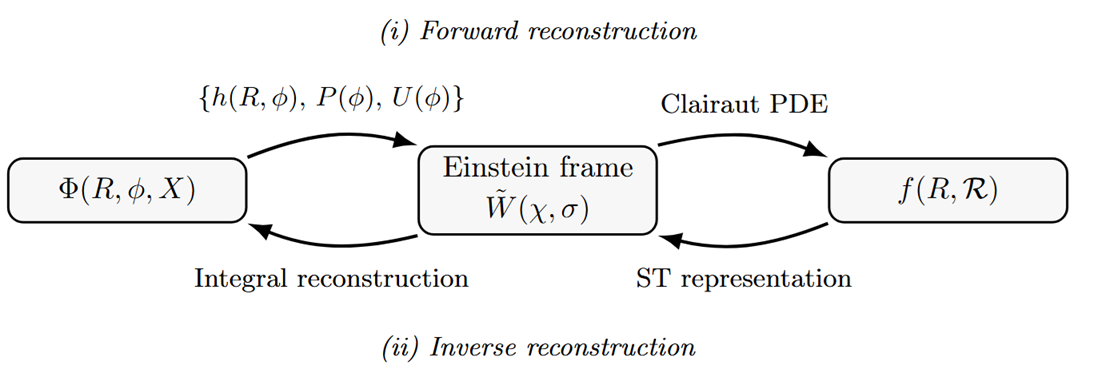}
  \caption{Schematic reconstruction between the regular Einstein-frame sectors of \(\Phi(R,\phi,X)\) theories linear in \(X\) and GH gravity. Path~\textit{(i)} denotes the forward Clairaut reconstruction, while path~\textit{(ii)} denotes the inverse reconstruction.}
    \label{fig:correspondence}
\end{figure}

It is important to stress that the correspondence is not a global equivalence
between arbitrary Jordan-frame actions. Rather, it is a local equivalence
between regular scalar--tensor branches that induce the same Einstein-frame
kinetic form and scalar potential. In the \(\Phi(R,\phi,X)\) sector, the
regularity conditions are
\begin{equation}
    h_{,R}>0,
    \qquad
    h_{,RR}\neq0,
    \qquad
    P(\phi)>0,
    \label{eq:Phi-regularity-conditions}
\end{equation}
whereas in the generalized hybrid sector they are
\begin{equation}
    f_{,R}+f_{,\mathcal R}>0,
    \qquad
    f_{,\mathcal R}<0,
    \qquad
    \Delta_f\equiv
    f_{,RR}f_{,\mathcal R\mathcal R}
    -
    f_{,R\mathcal R}^{\,2}
    \neq0 .
    \label{eq:GH-regularity-conditions}
\end{equation}
The first two conditions in Eq.~\eqref{eq:GH-regularity-conditions} ensure that
the Einstein-frame conformal factor is positive and that the field
\(\sigma\) is real, while the last condition ensures the local invertibility of
the Legendre map defining the scalar--tensor representation.

\subsection{\texorpdfstring{From $\Phi(R,\phi,X)$ to $f(R,\mathcal R)$}{From Phi to f}}
\label{subsec:Phi-to-f-theorem}

We first consider the reconstruction of GH representatives from a regular
\(\Phi(R,\phi,X)\) theory. The result is local: it applies on open domains where the scalar--tensor representations of both theories are
well defined and the corresponding field redefinitions are invertible.

\begin{theorem}[Forward reconstruction]
\label{thm:Phi-to-f}
Let a theory of the form
\begin{equation}
    \Phi(R,\phi,X)=h(R,\phi)+P(\phi)X-U(\phi)
\end{equation}
satisfy, on an open domain, the regularity conditions 
\begin{equation}
    h_{,R}>0,\qquad h_{,RR}\neq0,\qquad P(\phi)>0 .
    \label{eq:Phi-regularity-forward}
\end{equation}
Let \(\tilde W(\chi,\sigma)\) denote the corresponding Einstein-frame
potential. 

A GH function \(f(R,\mathcal R)\) defines, on a regular branch, the same vacuum Einstein-frame two-scalar sector if it satisfies

\begin{equation}
f(R,\mathcal R)
=
R f_{,R}
+\mathcal R f_{,\mathcal R}
-
2\kappa^2
\left(f_{,R}+f_{,\mathcal R}\right)^2
\tilde W
\left(
\sqrt{\frac{3}{2\kappa^2}}
\ln\left(f_{,R}+f_{,\mathcal R}\right),
\frac{\sqrt6}{\kappa}
\sqrt{-f_{,\mathcal R}}
\right)
\label{eq:clairaut-forward}
\end{equation}
on a domain where
\begin{equation}
    f_{,R}+f_{,\mathcal R}>0,
    \qquad
    f_{,\mathcal R}<0,
    \qquad
    f_{,RR}f_{,\mathcal R\mathcal R}
    -
    f_{,R\mathcal R}^{\,2}
    \neq0 .
    \label{eq:GH-regularity-forward}
\end{equation}
\end{theorem}

\begin{proof}
On the regular \(\Phi(R,\phi,X)\) branch, the theory can be written in
the Einstein frame as
\begin{equation}
S_E=
\int d^4x\sqrt{-\tilde g}
\left[
\frac{\tilde R}{2\kappa^2}
-\frac{1}{2}\tilde g^{\mu\nu}
\nabla_\mu\chi\nabla_\nu\chi
-\frac{1}{2}e^{-\sqrt{\frac{2}{3}}\kappa\chi}
\tilde g^{\mu\nu}\nabla_\mu\sigma\nabla_\nu\sigma
-\tilde W(\chi,\sigma)
\right].
\end{equation}
On the regular branch of GH gravity, the
Einstein-frame variables are
\begin{equation}
    \chi=
    \sqrt{\frac{3}{2\kappa^2}}\ln{(\varphi-\psi)},
    \qquad
    \sigma=\frac{\sqrt6}{\kappa}\sqrt{\psi},
\end{equation}
with
\begin{equation}\label{scalar-fields-R}
    \varphi=f_{,R},
    \qquad
    \psi=-f_{,\mathcal R}.
\end{equation}
Since the kinetic terms then have the same functional form, the reconstruction
condition reduces to the equality of the Einstein-frame potentials:
\begin{equation}\label{potencial-W-V}
      \frac{1}{2\kappa^2(\varphi-\psi)^2}
   V(\varphi,\psi)=\tilde{W}\!\left(
   \sqrt{\tfrac{3}{2\kappa^2}}\ln(\varphi-\psi),
   \tfrac{\sqrt{6}}{\kappa}\sqrt{\psi}
   \right).
\end{equation}
The scalar potential \(V(\varphi,\psi)\) is related to the original hybrid
function \(f(R,\mathcal R)\) by the Legendre transform
\eqref{potencial-V-hibrida}. Substituting the relations
\eqref{scalar-fields-R} into Eq.~\eqref{potencial-W-V} and using
Eq.~\eqref{potencial-V-hibrida}, one obtains Eq.~\eqref{eq:clairaut-forward}.

Conversely, if a function \(f(R,\mathcal R)\) satisfies Eq.~\eqref{eq:clairaut-forward} on a domain where
Eq.~\eqref{eq:GH-regularity-forward} holds, then the Legendre map is
locally invertible and the hybrid scalar--tensor representation is
regular. The Einstein-frame kinetic sector has the same field-space
metric as the \(\Phi(R,\phi,X)\) branch, and the potentials coincide by
construction. Therefore the two vacuum Einstein-frame scalar sectors are
locally identical.
\end{proof}

\subsection{Degenerate and singular Clairaut branches}
\label{subsec:Clairaut-branches}

Equation~\eqref{eq:clairaut-forward} is a first-order Clairaut-type
partial differential equation. Introducing
\begin{equation}
    p=f_{,R},
    \qquad
    q=f_{,\mathcal R},
\end{equation}
it can be written as
\begin{equation}
    f=Rp+\mathcal R q+G(p,q),
\end{equation}
where
\begin{equation}
G(p,q)
=
-2\kappa^2(p+q)^2
\tilde W
\left[
\sqrt{\frac{3}{2\kappa^2}}\ln(p+q),
\frac{\sqrt6}{\kappa}\sqrt{-q}
\right].
\end{equation}

\begin{proposition}[Degeneracy of the complete Clairaut integral]
\label{prop:linear-clairaut-degenerate}
The complete linear integral of Eq.~\eqref{eq:clairaut-forward} does
not define a regular two-scalar generalized hybrid branch.
\end{proposition}

\begin{proof}
The complete Clairaut integral is obtained by assigning constant values
to the first derivatives,
\begin{equation}
    p=a,
    \qquad
    q=b,
\end{equation}
which gives
\begin{equation}
f^{(l)}(R,\mathcal R)
=
aR+b\mathcal R
+
G(a,b).
\label{eq:linear-clairaut-integral}
\end{equation}
Since \(a\) and \(b\) are constants, one has
\begin{equation}
    f_{,RR}=0,
    \qquad
    f_{,\mathcal R\mathcal R}=0,
    \qquad
    f_{,R\mathcal R}=0 .
\end{equation}
Therefore,
\begin{equation}
    \Delta_f
    =
    f_{,RR}f_{,\mathcal R\mathcal R}
    -
    f_{,R\mathcal R}^{\,2}
    =
    0 .
\end{equation}
The Legendre map is not locally invertible, and the corresponding
hybrid scalar--tensor representation cannot describe a regular
two-scalar branch.
\end{proof}

The remaining Clairaut candidate is the singular envelope of the family
\eqref{eq:linear-clairaut-integral}. It is determined, when it exists,
by
\begin{equation}
    \frac{\partial f^{(l)}}{\partial a}=0,
    \qquad
    \frac{\partial f^{(l)}}{\partial b}=0 .
    \label{envelope-conditions}
\end{equation}
Thus, any regular GH model reconstructed from
a \(\Phi(R,\phi,X)\) model must come from the singular envelope and must
still be checked against the full set of conditions
\eqref{eq:GH-regularity-forward}.

\subsection{\texorpdfstring{Inverse reconstruction: from $f(R,\mathcal R)$ to $\Phi(R,\phi,X)$}{Inverse reconstruction}}
\label{subsec:f-to-Phi-theorem}

We now consider the inverse problem. Starting from a regular GH theory, or equivalently from its Einstein-frame potential
\(\tilde W_f(\chi,\sigma)\), we reconstruct a family of
\(\Phi(R,\phi,X)\) theories that share the same Einstein-frame scalar sector.
This reconstruction is local and depends on the choice of scalar-field
parametrization in the \(\Phi(R,\phi,X)\) representation.

\begin{theorem}[Inverse reconstruction]
\label{thm:f-to-Phi}
Let \(f(R,\mathcal R)\) be a generalized hybrid metric--Palatini theory
satisfying the conditions
\begin{equation}
    f_{,R}+f_{,\mathcal R}>0,
    \qquad
    f_{,\mathcal R}<0,
    \qquad
    f_{,RR}f_{,\mathcal R\mathcal R}
    -
    f_{,R\mathcal R}^{\,2}
    \neq0 .
    \label{eq:GH-regularity-inverse}
\end{equation}
Let \(\tilde W_f(\chi,\sigma)\) be its Einstein-frame potential. Choose a
locally invertible scalar parametrization
\begin{equation}
    \sigma=\mathcal J(\phi),
    \qquad
    \mathcal J'(\phi)=\frac{\sqrt{P(\phi)}}{\sqrt2\kappa},
    \qquad
    P(\phi)>0 .
    \label{eq:J-inverse}
\end{equation}
Define
\begin{equation}
\xi(\gamma,\phi)
=
\partial_\gamma
\left[
2\kappa^2\gamma^2
\tilde W_f
\left(
\sqrt{\frac{3}{2\kappa^2}}\ln\gamma,
\mathcal J(\phi)
\right)
\right].
\label{eq:xi-inverse}
\end{equation}
If Eq.~\eqref{eq:xi-inverse} can be locally inverted as
\begin{equation}
    \gamma=\gamma(\xi,\phi),
    \label{eq:gamma-inverse}
\end{equation}
on a domain where   
\begin{equation}
    \gamma>0,
    \qquad
    \partial_\xi\gamma\neq0,
\end{equation}
then
\begin{equation}
\Phi(R,\phi,X)=\int^R \gamma(\bar R,\phi)\,d\bar R
+P(\phi)X\label{eq:Phi-inverse}
\end{equation}
defines a \(\Phi(R,\phi,X)\) representative whose regular
Einstein-frame scalar sector coincides with the original hybrid branch.
Different choices of \(\mathcal J(\phi)\), or equivalently of
\(P(\phi)\), generate an equivalence class of such representatives.
\end{theorem}

\begin{proof}
Let the reconstructed \(\Phi(R,\phi,X)\) theory be written as
\begin{equation}
    \Phi(R,\phi,X)=h(R,\phi)+P(\phi)X-U(\phi).
\end{equation}
Following the Einstein-frame construction of
Sec.~\ref{sec:EF-Representation of GST}, introduce the auxiliary field \(\xi\)
and define
\begin{equation}
    \gamma=h_{,\xi}(\xi,\phi).
\end{equation}
The Einstein-frame potential generated by this \(\Phi(R,\phi,X)\) theory can be
written in terms of \((\gamma,\phi)\) as
\begin{equation}
\tilde W(\gamma,\phi)
=
\frac{1}{2\kappa^2}\gamma^{-2}
\left[
\gamma\,\xi(\gamma,\phi)
-
h(\xi(\gamma,\phi),\phi)
+
U(\phi)
\right].
\label{eq:W-Phi-gamma-proof}
\end{equation}
Requiring this potential to coincide with the prescribed hybrid potential gives
\begin{equation}
\gamma\,\xi(\gamma,\phi)
-
h(\xi(\gamma,\phi),\phi)
+
U(\phi)
=
2\kappa^2\gamma^2
\tilde W_f
\left(
\sqrt{\frac{3}{2\kappa^2}}\ln\gamma,
\mathcal J(\phi)
\right).
\label{eq:matching-inverse-proof}
\end{equation}
Differentiating Eq.~\eqref{eq:matching-inverse-proof} with respect to
\(\gamma\), and using \(h_{,\xi}=\gamma\), gives
Eq.~\eqref{eq:xi-inverse}. If this relation can be inverted locally, then
\begin{equation}
    h_{,\xi}(\xi,\phi)=\gamma(\xi,\phi).
\end{equation}
Therefore,
\begin{equation}
    h(\xi,\phi)=\int^\xi \gamma(\bar\xi,\phi)\,d\bar\xi+Q(\phi),
    \label{eq:h-inverse-proof}
\end{equation}
where \(Q(\phi)\) is an arbitrary integration function. Since the
Einstein-frame potential depends on \(h\) and \(U\) only through the combination
\(-h+U\), the function \(Q(\phi)\) can be absorbed into a redefinition of
\(U(\phi)\). Choosing \(Q(\phi)=U(\phi)\), and replacing \(\xi\) by \(R\), one
obtains Eq.~\eqref{eq:Phi-inverse}.
\end{proof}

Thus, the freedom in \(P(\phi)\), or equivalently in the choice of \(\sigma=\mathcal J(\phi)\), labels different Jordan-frame representatives
of the same Einstein-frame dynamics. The inverse map therefore assigns to each regular hybrid branch an equivalence class of \(\Phi(R,\phi,X)\) theories, rather than a unique representative.

As a consequence, the reconstruction induces a local map between vacuum
solution spaces. A vacuum solution of the \(\Phi(R,\phi,X)\) theory or of the
reconstructed GH theory can be mapped to the common Einstein-frame two-scalar system. Applying the corresponding inverse conformal transformation and scalar-field redefinitions then gives a vacuum solution in the other formulation. The equivalence holds only within the regular sectors
connected by the reconstruction.

\section{Applications}
\label{sec:Applications}

We now illustrate the reconstruction framework with explicit examples. The aim is not only to display particular solutions of the reconstruction equations, but also to show how the main structural features of the correspondence appear
in concrete models.

The examples are chosen to illustrate different aspects of the reconstruction. The quadratic \(\Phi(R,\phi,X)\) models show explicitly how the singular Clairaut envelope generates non-degenerate hybrid representatives and how the
regularity domain restricts the admissible branches. The scalaron--Higgs example demonstrates that a known two-field inflationary Einstein-frame sector can be embedded in GH gravity on a regular branch. The inverse examples show how prescribed hybrid potentials used in cosmological applications are translated into families of
\(\Phi(R,\phi,X)\) actions, with the hybrid parameters fixing the powers of the curvature and scalar-field dependence.

\subsection{\texorpdfstring{Quadratic \(\Phi(R,\phi,X)\) models}{Quadratic Phi models}}
\label{subsec:Quadratic-Phi-models}

We begin with the subclass of \(\Phi(R,\phi,X)\) theories whose dependence on
the metric Ricci scalar is quadratic,
\begin{equation}
\Phi(R,\phi,X)
=
\mathcal F(\phi)R
+
\frac{1}{6M^2(\phi)}R^2
+
k_0X
-
U(\phi),
\label{quadratic-R}
\end{equation}
where \(k_0>0\). This corresponds to
\begin{equation}
    P(\phi)=k_0,
    \qquad
    h(R,\phi)
    =
    \mathcal F(\phi)R
    +
    \frac{1}{6M^2(\phi)}R^2 .
\end{equation}
The functions \(\mathcal F(\phi)\) and \(M^2(\phi)\) generalize, respectively,
the effective non-minimal coupling to \(R\) and the scalaron mass scale. For
constant \(\mathcal F\), constant \(M\), and fixed \(\phi\), the curvature
sector reduces to the Starobinsky form. When \(\phi\) is dynamical,
Eq.~\eqref{quadratic-R} describes a two-field deformation of the scalaron
sector, of the type used in Starobinsky-like and scalaron--Higgs constructions
\cite{Sho2015,Dhimiter2020,Anirudh2020}.

The model belongs to the regular \(\Phi(R,\phi,X)\) branch provided
\begin{equation}
    h_{,R}
    =
    \mathcal F(\phi)
    +
    \frac{R}{3M^2(\phi)}
    >0,
    \qquad
    h_{,RR}
    =
    \frac{1}{3M^2(\phi)}
    \neq0,
    \qquad
    k_0>0 .
    \label{eq:quadratic-Phi-regularity}
\end{equation}

The Einstein-frame scalar fields are
\begin{equation}
    \sigma=\sqrt{\frac{k_0}{2\kappa^2}}\,\phi,
    \label{integral-X}
\end{equation}
and
\begin{equation}
\chi = \sqrt{\frac{3}{2\kappa^2}}
    \ln\left[\mathcal F(\phi)+
    \frac{R}{3M^2(\phi)}\right].
    \label{campo-chi-X}
\end{equation}
The corresponding Einstein-frame potential is
\begin{equation}
\tilde W(\chi,\sigma)=
\frac{1}{2\kappa^2}e^{-2\sqrt{\frac23}\kappa\chi}
\left[\frac{3}{2}M^2\!\Big(\sqrt{\tfrac{2\kappa^2}{k_0}}\sigma\Big)\left(e^{\sqrt{\frac23}\kappa\chi}
-\mathcal F\!\Big(\sqrt{\tfrac{2\kappa^2}{k_0}}\sigma\Big)\right)^2+
U\!\Big(\sqrt{\tfrac{2\kappa^2}{k_0}}\sigma\Big)
\right].
\label{potential-quadratic-R}
\end{equation}

According to Theorem~\ref{thm:Phi-to-f}, any regular GH representative
associated with this Einstein-frame scalar sector must satisfy the Clairaut
equation \eqref{eq:clairaut-forward}, with \(\tilde W(\chi,\sigma)\) given by
Eq.~\eqref{potential-quadratic-R}. We shall consider two choices of the functions \(\mathcal F(\phi)\), \(M^2(\phi)\), and \(U(\phi)\). The first is motivated by wormhole configurations and is used to identify the corresponding gravitational--scalar hybrid branch. The second realizes the scalaron--Higgs two-field sector within GH gravity.

\subsubsection{Wormhole configuration}
\label{subsec:Wormhole-configuration}

As a first application of Theorem~\ref{thm:Phi-to-f}, we consider a quadratic
model motivated by wormhole configurations previously studied in
\(\Phi(R,\phi,X)\) gravity \cite{Malik-2021}. Since these configurations involve matter sources,
the purpose of the present example is not to map the complete
matter-supported geometry. Instead, we use the corresponding scalar sector as
input for the reconstruction and determine the regular GH model associated with it. 

We consider the
quadratic model
\begin{equation}
    M^2(\phi)=\frac{1}{6\mu},
    \qquad
    \mathcal F(\phi)=1,
    \qquad
    U(\phi)=U_0\phi^m,
    \qquad
    k_0=1 .
    \label{eq:symmetric-model-choice}
\end{equation}
This choice gives
\begin{equation}\label{explicit-Phi}
    \Phi(R,\phi,X)=R+\mu R^2+X-U_0\phi^m .
\end{equation}
The model is the simplest nontrivial member of the quadratic class \eqref{quadratic-R}. The \(R^2\) correction introduces the scalaron degree of
freedom, while the canonical scalar sector with a power-law potential preserves the structure required for the two-field Einstein-frame representation. It is
also sufficiently simple to allow an explicit Clairaut reconstruction. Thus, it provides a useful benchmark for testing
whether a known \(\Phi(R,\phi,X)\) gravitational-scalar sector admitting nontrivial configurations can be assigned a regular GH model. 

Its Einstein-frame potential \eqref{potential-quadratic-R} is given by
\begin{equation}
\tilde W(\chi,\sigma)
=
\frac{1}{2\kappa^2}
e^{-2\sqrt{\frac23}\kappa\chi}
\left[
\frac{1}{4\mu}
\left(
e^{\sqrt{\frac23}\kappa\chi}-1
\right)^2
+
U_0
\left(
\sqrt{2\kappa^2}\,\sigma
\right)^m
\right].
\label{potential-4}
\end{equation}

Substitution of Eq.~\eqref{potential-4} into the forward reconstruction equation
\eqref{eq:clairaut-forward} gives
\begin{equation}
 \frac{1}{4\mu}\left(f_{,R}+f_{,\mathcal R}-1\right)^2+U_0\left(2\sqrt{3}\sqrt{-f_{,\mathcal R}}\right)^m- Rf_{,R}-\mathcal R f_{,\mathcal R}
 +f(R,\mathcal R)=0 .
 \label{eq:clairaut-symmetric}
\end{equation}
The complete linear integral is
\begin{equation}
f^{(l)}(R,\mathcal R)=aR+b\mathcal R
-\frac{1}{4\mu}(a+b-1)^2-U_0
\left(2\sqrt{3}\sqrt{-b}\right)^m ,
\label{general-sol-1}
\end{equation}
where \(a\) and \(b\) are constants. The singular branch follows from the
envelope conditions, Eq.~\eqref{envelope-conditions},
\begin{equation}
R-\frac{1}{2\mu}(a+b-1)=0,
\qquad
\mathcal R-\frac{1}{2\mu}(a+b-1)
+
\frac{mU_0}{2}\,(2\sqrt{3})^m(-b)^{\frac{m}{2}-1}
=0 .
\label{eq:envelope-symmetric}
\end{equation}
Solving these equations for \(a\) and \(b\), and substituting back into
Eq.~\eqref{general-sol-1}, gives the hybrid representative
\begin{equation}
f(R,\mathcal R)
=
R+\mu R^2
+
A_m(R-\mathcal R)^{\frac{m}{m-2}},
\label{hybrid-wormhole}
\end{equation}
where
\begin{equation}
 A_m=
 \frac{m-2}
 {m\left(\frac{mU_0}{2} (2\sqrt{3})^m\right)^{\frac{2}{m-2}}}.
 \label{eq:Am-definition}
\end{equation}
The Hessian determinant is
\begin{equation}
\Delta_f
=\frac{4 \mu m A_m}{(m-2)^2}(R-\mathcal R)^{\tfrac{4-m}{m-2}}.
\label{eq:Hessian-symmetric}
\end{equation}
Therefore, this reconstructed model defines a regular hybrid branch on the
domain
\begin{equation}
\frac{4 \mu m A_m }{(m-2)^2}(R-\mathcal R)^{\tfrac{4-m}{m-2}}\neq0 .
    \label{eq:regular-domain-symmetric}
\end{equation}
In particular, the branch excludes \(R=\mathcal R\), \(m=0\), and \(m=2\), and
requires a real choice of the power \((R-\mathcal R)^{\tfrac{4-m}{m-2}}\). For non-integer
\(\tfrac{4-m}{m-2}\), this means that a real branch of \((R-\mathcal R)^{\tfrac{4-m}{m-2}}\) must be fixed consistently with the positivity condition \(\psi>0\).

The scalar profile used below is taken from matter-supported wormhole
solutions in \(\Phi(R,\phi,X)\) gravity. In the present reconstruction it is
used only to identify the corresponding gravitational--scalar branch. A full
mapping of the complete wormhole solution would require specifying the matter
action and its coupling in both Jordan-frame formulations. In the wormhole solutions considered in \cite{Malik-2021}, the scalar field can be written as
\begin{equation}
    \phi(r)=a_0\left(\frac{b_0}{r}\right)^{\sigma_1},
    \label{eq:phi-profile}
\end{equation}
where \(a_0\), \(b_0\), and \(\sigma_1\) are constants. Since \(k_0=1\), the
Einstein-frame scalar \(\sigma\) is
\begin{equation}
    \sigma(r)
    =
    \frac{1}{\sqrt{2}\kappa}
    a_0\left(\frac{b_0}{r}\right)^{\sigma_1},
    \label{eq:sigma-profile-symmetric}
\end{equation}
where we have set the integration constant in \(\mathcal J(\phi)\) to zero.
The hybrid scalar \(\psi\) is then fixed by
\begin{equation}
    \psi(r)
    =
    \frac{\kappa^2\sigma^2(r)}{6}.
    \label{eq:psi-profile-symmetric}
\end{equation}
Using $\psi=-f_{,\mathcal R}$, the Palatini scalar curvature along the
corresponding matched scalar branch is reconstructed as
\begin{equation}
    \mathcal R(r)
    =
    R(r)
    -
    \left[
    \frac{\kappa^2(m-2)\sigma^2(r)}
    {6mA_m}
    \right]^{\frac{m-2}{2}} .
    \label{eq:Rcal-profile-symmetric}
\end{equation}
Thus, once the Einstein-frame scalar profile is fixed, the Palatini curvature
scalar is not arbitrary but is determined algebraically by the reconstructed
hybrid branch.

A comment is in order concerning the metric Ricci scalar \(R\) appearing in the
reconstruction. The correspondence developed here is formulated at the level of
the Einstein-frame scalar sector. Therefore, the metric Ricci scalar in the
original \(\Phi(R,\phi,X)\) representative and the metric Ricci scalar in the
reconstructed hybrid representative should, in principle, be distinguished
before imposing any additional identification of the Jordan-frame geometries.

Let \(R_\Phi\) denote the metric Ricci scalar in the \(\Phi(R,\phi,X)\)
representative and \(R_{\rm H}\) the metric Ricci scalar in the hybrid
representative. Matching the conformal scalar gives the general condition
\begin{equation}
    h_{,R}(R_\Phi,\phi)
    =
    f_{,R}(R_{\rm H},\mathcal R)
    +
    f_{,\mathcal R}(R_{\rm H},\mathcal R).
    \label{eq:general-conformal-matching}
\end{equation}
For the reconstructed hybrid branch \eqref{hybrid-wormhole}, this condition
becomes
\begin{equation}
    1+2\mu R_\Phi
    =
    1+2\mu R_{\rm H},
\end{equation}
and hence
\begin{equation}
    R_\Phi=R_{\rm H}.
    \label{eq:Rphi-RH-relation}
\end{equation}
Thus, in this example, the matching of the Einstein-frame conformal scalar is
consistent with identifying the same Levi--Civita Ricci scalar in both
Jordan-frame representatives. The remaining Palatini scalar \(\mathcal R\) is
then fixed algebraically by Eq.~\eqref{eq:Rcal-profile-symmetric}.

This example shows that the Clairaut reconstruction does more than generate a
formal function \(f(R,\mathcal R)\). The reconstructed function fixes the relation \(\psi=-f_{,\mathcal R}\), so that, once the matched Einstein-frame scalar profile is specified, the Palatini scalar \(\mathcal R\) is determined algebraically on the corresponding hybrid branch.

\subsubsection{Hybrid realization of the scalaron--Higgs sector}
\label{subsec:scalaron-higgs-hybrid}

As a second application of the forward reconstruction theorem, we consider the scalaron--Higgs sector \cite{Anirudh2020}. This example is physically motivated
because it combines the \(R^2\) correction of Starobinsky-type inflation with the non-minimal scalar-curvature coupling characteristic of Higgs-inflation
models. In the Einstein frame, these ingredients generate an interacting two-field scalar sector with a nontrivial valley structure, making the model a useful test of the reconstruction map.

In the present notation, this model is obtained from the quadratic class
\eqref{quadratic-R} by choosing
\begin{equation}
    M^2(\phi)=M_0^2,
    \qquad
    \mathcal F(\phi)=m_0^2+\alpha_0\phi^2,
    \qquad
    U(\phi)=\frac{\lambda}{2}(\phi^2-\nu^2)^2,
    \qquad
    k_0=2 .
    \label{eq:Scalaron-Higgs-choice}
\end{equation}
The regularity conditions in the \(\Phi(R,\phi,X)\) representative require
\begin{equation}
    m_0^2+\alpha_0\phi^2+\frac{R}{3M_0^2}>0,
    \qquad
    M_0^2\neq0,
    \qquad
    k_0=2>0 .
    \label{eq:Scalaron-Higgs-Phi-regularity}
\end{equation}
The Einstein-frame fields are
\begin{equation}
    \sigma=\frac{\phi}{\kappa},
    \qquad
    \chi=
    \sqrt{\frac{3}{2\kappa^2}}
    \ln
    \left(
    m_0^2+\alpha_0\phi^2+\frac{R}{3M_0^2}
    \right).
    \label{eq:Scalaron-Higgs-fields}
\end{equation}
The corresponding Einstein-frame potential is
\begin{equation}
\tilde W(\chi,\sigma)
=
\frac{1}{4\kappa^2}
e^{-2\sqrt{\frac23}\kappa\chi}
\left[
3M_0^2
\left(
m_0^2+\alpha_0\kappa^2\sigma^2
-
e^{\sqrt{\frac23}\kappa\chi}
\right)^2
+
\lambda
\left(
\kappa^2\sigma^2-\nu^2
\right)^2
\right].
\label{potential-quadratic-R-a}
\end{equation}

Using Eq.~\eqref{eq:clairaut-forward}, or equivalently matching the hybrid
potential \(V(\varphi,\psi)\), one obtains
\begin{equation}
V(\varphi,\psi)
=
\frac{3}{2}M_0^2
\left[
m_0^2-\varphi+(1+6\alpha_0)\psi
\right]^2
+
\frac{\lambda}{2}
\left(
6\psi-\nu^2
\right)^2 .
\label{eq:V-Scalaron-Higgs}
\end{equation}
The corresponding Clairaut equation is
\begin{equation}
\frac{3}{2}M_0^2
\left(m_0^2-f_{,R}-(1+6\alpha_0)f_{,\mathcal R}
\right)^2+\frac{\lambda}{2}
\left(6f_{,\mathcal R}+\nu^2\right)^2
-Rf_{,R}-\mathcal R f_{,\mathcal R}
+f(R,\mathcal R)=0 .
\label{clairaut-U2}
\end{equation}
Introducing
\begin{equation}
    \eta=\frac{1+6\alpha_0}{6},
    \label{eq:eta-definition}
\end{equation}
the singular Clairaut envelope yields
\begin{equation}
f(R,\mathcal R)=
\left(\eta\nu^2+m_0^2\right)R
+\left(\frac{\eta^2}{2\lambda}
+\frac{1}{6M_0^2}\right)R^2
-\frac{\eta}{6\lambda}R\mathcal R
-\frac{\nu^2}{6}\mathcal R
+\frac{1}{72\lambda}\mathcal R^2,
\label{Scalaron-Higgs-Hibrida}
\end{equation}
with \(\lambda\neq0\) and \(M_0^2\neq0\).

The regularity of the reconstructed hybrid branch can be checked explicitly.
From Eq.~\eqref{Scalaron-Higgs-Hibrida}, one finds
\begin{equation}
\Delta_f=f_{,RR}f_{,\mathcal R\mathcal R}
-f_{,R\mathcal R}^{\,2}=\frac{1}{108\lambda M_0^2}, \qquad     \lambda M_0^2\neq0 .
\label{eq:Hessian-Scalaron-Higgs}
\end{equation}
Thus, for the usual scalaron--Higgs parameter range
\(\lambda>0\) and \(M_0^2>0\), the Hessian determinant is positive.

The two scalar variables of the hybrid representation are
\begin{equation}
\psi=-f_{,\mathcal R}=
\frac{\eta}{6\lambda}R
+\frac{\nu^2}{6}-\frac{1}{36\lambda}\mathcal R,
\label{eq:psi-Scalaron-Higgs}
\end{equation}
and
\begin{equation}
\vartheta=f_{,R}+f_{,\mathcal R}
=m_0^2+\frac{R}{3 M_0^2}+\frac{(6 \eta -1) }{36\lambda }\left(6 \lambda  \nu ^2+6 \eta  R-\mathcal R\right)
\label{eq:vartheta-Scalaron-Higgs}
\end{equation}
Accordingly, the regular Einstein-frame branch is defined by
\begin{equation}
    \psi>0, \qquad \vartheta>0, \qquad
    \Delta_f\neq0 .
    \label{eq:regular-domain-Scalaron-Higgs}
\end{equation}
Equivalently,
\begin{equation}
    \frac{\eta}{6\lambda}R
+\frac{\nu^2}{6}-\frac{1}{36\lambda}\mathcal R>0,
    \qquad
  f_{,R}+f_{,\mathcal R}
=m_0^2+\frac{R}{3 M_0^2}+\frac{(6 \eta -1) }{36\lambda }\left(6 \lambda  \nu ^2+6 \eta  R-\mathcal R\right) >0 .
    \label{eq:regular-domain-Scalaron-Higgs-explicit}
\end{equation}

Equation~\eqref{Scalaron-Higgs-Hibrida} is therefore an explicit GH
representative of the scalaron--Higgs two-field Einstein-frame sector. Since
the reconstructed branch has the same kinetic form and the same
Einstein-frame potential \eqref{potential-quadratic-R-a}, it realizes the
same regular scalar sector within GH gravity.
The purpose of this example is structural: it shows that a known two-field inflationary sector can be embedded in a regular generalized hybrid branch. The inequalities \eqref{eq:regular-domain-Scalaron-Higgs}--\eqref{eq:regular-domain-Scalaron-Higgs-explicit} specify the domain in which this realization defines a genuine two-scalar
hybrid representation.

\subsection{Inverse reconstruction from prescribed hybrid potentials}
\label{subsec:inverse-prescribed-potentials}

We now illustrate the inverse reconstruction theorem \ref{thm:f-to-Phi}. In this direction, the input is a potential associated with a regular generalized hybrid branch.  The output is not a unique \(\Phi(R,\phi,X)\) theory, but a
family of representatives parametrized by the scalar-field \(\sigma=\mathcal J(\phi)\).

We consider the class of hybrid Einstein-frame potentials
\begin{equation}
\tilde W_f(\chi,\sigma)
=
\frac{1}{2\kappa^2}
\left[
w_0+w_1(\kappa\sigma)^{2p}
\right]
e^{-q\sqrt{\frac{2}{3}}\kappa\chi},
\label{generalized-exp-potential}
\end{equation}
where \(w_0\), \(w_1\), \(p\), and \(q\) are constants. This class includes potentials that have appeared in dynamical-systems analyses and cosmological
applications of GH gravity \cite{tamanini,PauloHybrid2020,Rosa2017}.

Using the inverse reconstruction theorem, we write
\begin{equation}
    \gamma=e^{\sqrt{\frac{2}{3}}\kappa\chi},
    \qquad
    \sigma=\mathcal J(\phi),
\end{equation}
so that Eq.~\eqref{generalized-exp-potential} becomes
\begin{equation}
\tilde W_f(\gamma,\phi)
=
\frac{1}{2\kappa^2}
\left[
w_0+w_1\bigl(\kappa\mathcal J(\phi)\bigr)^{2p}
\right]\gamma^{-q}.
\label{W-calE-1}
\end{equation}
Substitution into Eq.~\eqref{eq:xi-inverse} gives
\begin{equation}
\xi
=
(2-q)
\left[
w_0+w_1\bigl(\kappa\mathcal J(\phi)\bigr)^{2p}
\right]
\gamma^{1-q}.
\label{eq:xi-general-potential}
\end{equation}
For \(q\neq1\) and \(q\neq2\), this relation can be locally inverted as
\begin{equation}
\gamma(\xi,\phi)
=
\left[
\frac{\xi}
{(2-q)
\left(
w_0+w_1\bigl(\kappa\mathcal J(\phi)\bigr)^{2p}
\right)}
\right]^{\frac{1}{1-q}}.
\label{eq:gamma-general-potential}
\end{equation}
Therefore, the reconstructed family of \(\Phi(R,\phi,X)\) theories is
\begin{equation}
\Phi(R,\phi,X)
=
\mathcal C_q
R^{\frac{2-q}{1-q}}
\left[
w_0+w_1\bigl(\kappa\mathcal J(\phi)\bigr)^{2p}
\right]^{-\frac{1}{1-q}}
+
P(\phi)X,
\label{Phi-general-final}
\end{equation}
where
\begin{equation}
    \mathcal C_q
    =
    (1-q)(2-q)^{-\frac{2-q}{1-q}},
    \qquad
    q\neq1,\qquad q\neq2.
    \label{eq:Cq-definition}
\end{equation}
The function \(P(\phi)\) remains arbitrary up to the local invertibility
condition \(P(\phi)>0\). Thus, a prescribed hybrid Einstein-frame potential
determines an equivalence class of \(\Phi(R,\phi,X)\) representatives rather
than a unique Jordan-frame action.

\subsubsection{Power-law exponential potential}
\label{subsec:power-law-exponential-potential}

As a first particular case of Eq.~\eqref{generalized-exp-potential}, we take
\begin{equation}
    w_0=0,
    \qquad
    2p=2q=s.
\end{equation}
The Einstein-frame potential then becomes
\begin{equation}
\tilde W_f(\chi,\sigma)
=
\frac{w_1}{2\kappa^2}
(\kappa\sigma)^s
e^{-s\sqrt{\frac{1}{6}}\kappa\chi}.
\label{W-w=0-p=2q}
\end{equation}
This coincides with the functional form of the interacting two-field potential
used in the dynamical-systems analysis of GH
gravity in \cite{tamanini}, with the identification \(s=\lambda\).

The inverse reconstruction formula \eqref{Phi-general-final} gives
\begin{equation}
\Phi(R,\phi,X)
=
\mathcal C_{s/2}
R^{\frac{4-s}{2-s}}
\left[
w_1
\bigl(\kappa\mathcal J(\phi)\bigr)^s
\right]^{-\frac{2}{2-s}}
+
P(\phi)X,
\qquad
s\neq2,4 .
\label{Phi-power-law-general}
\end{equation}
For a constant kinetic coupling \(P(\phi)=a_1>0\), and choosing the integration
constant in \(\mathcal J\) to vanish, one has
\begin{equation}
    \mathcal J(\phi)
    =
    \sqrt{\frac{a_1}{2\kappa^2}}\phi .
\end{equation}
In this case, Eq.~\eqref{Phi-power-law-general} reduces to
\begin{equation}
\Phi(R,\phi,X)
=
a_0
\phi^{\frac{2s}{s-2}}
R^{\frac{4-s}{2-s}}
+
a_1X,
\qquad
s\neq2,4 ,
\label{Phi-power-law-constant-P}
\end{equation}
where the constant \(a_0\) absorbs the numerical factors depending on
\(w_1\), \(a_1\), \(\kappa\), and \(s\).

This example shows how a parameter appearing in the hybrid Einstein-frame
dynamical system is translated into the powers of a reconstructed
\(\Phi(R,\phi,X)\) action. Writing
\begin{equation}
\Phi(R,\phi,X)= a_0\phi^{\alpha_\phi}R^{\alpha_R}
+ a_1X , \label{eq:Phi-power-law-alpha}
\end{equation}
the inverse reconstruction fixes
\begin{equation}
    \alpha_R
    =
    \frac{4-s}{2-s},
    \qquad
    \alpha_\phi
    =
    \frac{2s}{s-2}.
    \label{eq:alphaR-alphaphi}
\end{equation}
Thus, the parameter \(s\), which controls the fixed-point structure of the
vacuum two-scalar Einstein-frame system in the hybrid formulation, determines
both curvature and scalar-field powers in the reconstructed
\(\Phi(R,\phi,X)\) model. These powers are not independent, but
satisfy
\begin{equation}
    \alpha_\phi=4-2\alpha_R .
    \label{eq:alpha-relation}
\end{equation}

In the hybrid model \cite{tamanini}, the corresponding scaling points
have
\begin{equation}
    w_{\rm eff}
    =
    \frac{s^2}{18}-1 .
    \label{eq:weff-s}
\end{equation}
Therefore, for \(s>0\), the accelerating sector \(w_{\rm eff}<-1/3\)
corresponds to
\begin{equation}
    0<s<2\sqrt3 .
    \label{eq:accelerating-s-range}
\end{equation}
Through Eq.~\eqref{eq:alphaR-alphaphi}, this interval selects the corresponding
subclasses of \(\Phi(R,\phi,X)\) models. Since the inverse reconstruction is
singular at \(s=2\), the accelerating interval splits into two disjoint
branches:
\begin{equation}
    0<s<2:
    \qquad
    \alpha_R>2,
    \qquad
    \alpha_\phi<0,
    \label{eq:accelerating-branch-1}
\end{equation}
and
\begin{equation}
    2<s<2\sqrt3:
    \qquad
    \alpha_R<\frac{1-\sqrt3}{2},
    \qquad
    \alpha_\phi>3+\sqrt3 .
    \label{eq:accelerating-branch-2}
\end{equation}
At the level of the regular vacuum Einstein-frame branch, the image of the
hybrid accelerating sector is represented in the reconstructed
\(\Phi(R,\phi,X)\) family by two disconnected regions in the exponent space
\((\alpha_R,\alpha_\phi)\).

Selected regimes of the hybrid dynamical system correspond to definite powers
in the reconstructed \(\Phi(R,\phi,X)\) action. For example, the matter-like
condition \(w_{\rm eff}=0\) gives \(s=3\sqrt2\), while the radiation-like
condition \(w_{\rm eff}=1/3\) gives \(s=2\sqrt6\). The corresponding exponents
are summarized in Table~\ref{tab:hybrid-to-Phi-exponents}.

\begin{table}[h]
\centering
\begin{tabular}{c|c|c|c}
Hybrid regime & \(s\) & \(\alpha_R\) & \(\alpha_\phi\) \\
\hline
Accelerating scaling sector
&
\(0<s<2\sqrt3,\; s\neq2\)
&
\(\tfrac{4-s}{2-s}\)
&
\(\tfrac{2s}{s-2}\)
\\[2mm]
Effective dust-like scalar regime, \(w_{\rm eff}=0\)
&
\(3\sqrt2\)
&
\(\tfrac{5-3\sqrt2}{7}\)
&
\(\tfrac{6(3+\sqrt2)}{7}\)
\\[2mm]
Effective radiation-like scalar regime, \(w_{\rm eff}=1/3\)
&
\(2\sqrt6\)
&
\(\tfrac{4-\sqrt6}{5}\)
&
\(\tfrac{2(6+\sqrt6)}{5}\)\\ \hline
\end{tabular}
\caption{Effective fixed-point regimes of the
vacuum hybrid Einstein-frame two-scalar system into the powers of the
reconstructed \(\Phi(R,\phi,X)\) model \eqref{eq:Phi-power-law-alpha}.  }
\label{tab:hybrid-to-Phi-exponents}
\end{table}

For completeness, the same Einstein-frame potential also admits a generalized
hybrid representative through the forward reconstruction equation. The
corresponding singular Clairaut branch is
\begin{equation}
f(R,\mathcal R)
=
A_s
R^{\frac{4-s}{2}}
(R-\mathcal R)^{\frac{s}{2}},
\label{hybrid-f-Exp}
\end{equation}
with
\begin{equation}
A_s=\frac{4}{w_1(4-s)^2}
\left(\frac{4-s}{6s}\right)^{s/2}.
\label{eq:As-definition}
\end{equation}
For non-integer exponents, Eq.~\eqref{hybrid-f-Exp} is understood on a fixed
real branch.

The regular real branch of the reconstructed hybrid representative is defined by
\begin{equation}
    w_1\neq0,
    \qquad
    s\neq0,4,
    \qquad
    R\neq0,
    \qquad
    R\neq\mathcal R,
\end{equation}
together with
\begin{equation}
A_s(4-s)
R^{\frac{2-s}{2}}
(R-\mathcal R)^{\frac{s}{2}}>0,
\qquad
sA_s
R^{\frac{4-s}{2}}
(R-\mathcal R)^{\frac{s-2}{2}}>0 .
\end{equation}

The parameter \(s\) plays a dual role, fixing the slope of the exponential
interaction and the fixed-point structure in the hybrid Einstein-frame
description, while determining the powers \((\alpha_R,\alpha_\phi)\) in the
reconstructed \(\Phi(R,\phi,X)\) theory. In particular, the hybrid accelerating range is mapped to $0<s<2\sqrt3$ with $s\neq2$, where the exclusion \(s=2\) comes from the inverse \(\Phi\)-reconstruction. Thus the same Einstein-frame accelerating sector is represented in the \(\Phi(R,\phi,X)\) family by the two exponent branches identified above.

\subsubsection{Quadratic exponential potential}
\label{subsec:quadratic-exponential-potential}

As a second example of the inverse reconstruction, we consider the case
\(p=1\) in Eq.~\eqref{generalized-exp-potential}. The Einstein-frame potential
is then
\begin{equation}
\tilde W_f(\chi,\sigma)
=
\frac{1}{2\kappa^2}
\left[
w_0+w_1(\kappa\sigma)^2
\right]
e^{-q\sqrt{\frac{2}{3}}\kappa\chi}.
\label{quadratic-exp-potential}
\end{equation}
This form is motivated by dark-sector applications of GH gravity, where the two Einstein-frame scalars are used to
describe effective dark-energy and dark-matter components \cite{PauloHybrid2020}. Indeed, with the identifications
\begin{equation}
    q=\frac{\lambda}{2},
    \qquad
    w_0=2\kappa^2 V_a,
    \qquad
    w_1=\frac{2V_a}{\xi_a^2},
\end{equation}
Eq.~\eqref{quadratic-exp-potential} reproduces the scalar potential used in \cite{PauloHybrid2020}. In that model, the quadratic dependence on the
second scalar allows an oscillating-field interpretation as an effective
pressureless dark-matter component, while the exponential dependence on the
first scalar drives late-time dark-energy-like behavior. 

In the present reconstruction, we use only the two-scalar Einstein-frame potential as input. Therefore, the map below identifies the family of
\(\Phi(R,\phi,X)\) representatives associated with this scalar sector, but does not reconstruct the matter action or the matter-coupling prescription of
the cosmological model.

The same scalar-sector analysis gives the asymptotic relation
\begin{equation}
    w_{\rm eff}\to -1+\frac{2q^2}{9},
\end{equation}
so that the late-time accelerating sector corresponds to \(q\leq\sqrt3\).
The values \(q\sim1\) are particularly relevant for a standard-like sequence of
radiation, matter, and dark-energy domination, although the exact value
\(q=1\) is a singular point of the inverse reconstruction below.

Using Eq.~\eqref{Phi-general-final}, the associated family of
\(\Phi(R,\phi,X)\) theories is
\begin{equation}
\Phi(R,\phi,X)
=
\mathcal C_q
R^{\frac{2-q}{1-q}}
\left[
w_0+w_1
\bigl(\kappa\mathcal J(\phi)\bigr)^2
\right]^{-\frac{1}{1-q}}
+
P(\phi)X,
\qquad
q\neq1,2,
\label{Phi-Exp-quad}
\end{equation}
where
\begin{equation}
    \mathcal J'(\phi)=\frac{\sqrt{P(\phi)}}{\sqrt2\kappa},
    \qquad
    P(\phi)>0 .
\end{equation}
The reconstruction is understood on domains where the powers appearing in
Eq.~\eqref{Phi-Exp-quad} are real and the expression inside the square brackets
does not vanish. Thus, as in the general inverse construction, the potential
\eqref{quadratic-exp-potential} determines an equivalence class of
\(\Phi(R,\phi,X)\) representatives rather than a unique Jordan-frame action.

The dark-sector parameter \(q\) is therefore translated into the curvature
power
\begin{equation}
    \alpha_R=\frac{2-q}{1-q}
\end{equation}
of the reconstructed \(\Phi(R,\phi,X)\) representative. In particular, the
accelerating range \(q\leq\sqrt3\), excluding the singular point \(q=1\),
selects the corresponding subclasses within the family
\eqref{Phi-Exp-quad}. This should not be interpreted as a new dark-sector
dynamics, since the Einstein-frame potential is the same; rather, it identifies
the \(\Phi(R,\phi,X)\) representatives associated with the GH potential used in
the dark-sector construction.

For completeness, the same Einstein-frame potential also admits a GH representative through the forward reconstruction equation. The corresponding singular Clairaut branch is
\begin{equation}
f(R,\mathcal R)=R\left(\frac{R-\mathcal R}{6w_1}
\right)^{\frac{1}{2-q}}+\frac{w_0}{6w_1}
(\mathcal R-R),\qquad q\neq1,2,\qquad
w_1\neq0 .\label{hybrid-f-Exp-quad}
\end{equation}
For non-integer exponents, Eq.~\eqref{hybrid-f-Exp-quad} is understood on a
fixed real branch. The Hessian determinant of this representative is
\begin{equation}
\Delta_f=-
\frac{1}{(2-q)^2}
(6w_1)^{-\frac{2}{2-q}}
(R-\mathcal R)^{\frac{2(q-1)}{2-q}} .
\label{eq:Hessian-quad-exp}
\end{equation}
Hence the Legendre map is non-degenerate provided
\begin{equation}
    w_1\neq0,
    \qquad
    q\neq2,
    \qquad
    R\neq\mathcal R,
\end{equation}
together with the appropriate real branch choice for the power of \(R-\mathcal R\). In the present inverse reconstruction, the value \(q=1\) is
excluded by Eq.~\eqref{Phi-Exp-quad}, although it is not a degeneracy of the
Hessian in Eq.~\eqref{eq:Hessian-quad-exp}. The regular GH domain is further
restricted by the positivity conditions \(\psi>0\) and \(\vartheta>0\). 

In the limiting case \(w_1=0\), the potential reduces to a pure exponential in \(\chi\),
\begin{equation}
\tilde W_f(\chi,\sigma)
=
\frac{w_0}{2\kappa^2}
e^{-q\sqrt{\frac{2}{3}}\kappa\chi},
\end{equation}
and loses its dependence on \(\sigma\). This is the type of exponential Einstein-frame potential studied in the scalar--tensor formulation of GH gravity \cite{tamanini}. However, within the present reconstruction it does not define a regular two-scalar \(f(R,\mathcal R)\) branch. The reason is that the second scalar direction is
not fixed by the potential, so the corresponding Legendre map loses rank. This is consistent with the non-degeneracy requirement on the Hessian of \(f(R,\mathcal R)\) in the scalar--tensor equivalence. Accordingly, the case \(w_1=0\) may be studied as a degenerate scalar--tensor sector, but it is not
contained in the regular two-field \(f(R,\mathcal R)\) branch described by Eq.~\eqref{hybrid-f-Exp-quad}.

\section{Conclusions}
\label{sec:Conclusions}

We have formulated a local reconstruction between
\(\Phi(R,\phi,X)\) theories with linear dependence on \(X\) and GH gravity. The construction is performed in vacuum and at the level of regular scalar--tensor representation. On such branches, both theories admit the same Einstein-frame two-field kinetic geometry, so that the correspondence reduces to matching the scalar potentials together with the corresponding Legendre, positivity, and invertibility conditions.

In the forward direction, starting from a regular
\(\Phi(R,\phi,X)\) model, we derived a Clairaut-type equation for the corresponding GH function \(f(R,\mathcal R)\). The complete linear Clairaut integral was shown to be degenerate from the hybrid
scalar--tensor point of view, since the associated Hessian determinant vanishes. Therefore, it does not define a regular two-scalar hybrid representation. Non-degenerate hybrid representatives, when they exist, must
instead arise from the singular envelope and must still satisfy the regularity conditions of GH gravity. This provides a criterion for distinguishing admissible scalar--tensor branches from merely formal solutions of the reconstruction equation.

The inverse reconstruction gives a complementary result. A given regular
hybrid Einstein-frame potential determines a family, rather than a unique
representative, of \(\Phi(R,\phi,X)\) theories. The freedom is parametrized by
the choice of scalar-field redefinition, or equivalently by the kinetic
coupling of the original scalar field. Thus, different Jordan-frame
parametrizations may represent the same Einstein-frame two-scalar dynamics.

The examples make these results explicit. For quadratic
\(\Phi(R,\phi,X)\) models, the singular Clairaut envelope gives
non-degenerate generalized hybrid representatives, and the admissible domains
are fixed by the corresponding positivity and Hessian conditions. In the
wormhole-motivated case, the scalar profile was used only to identify the
associated gravitational--scalar branch; the Palatini scalar is then fixed
algebraically by the reconstructed hybrid function. This does not constitute a
map of the full matter-supported wormhole solution, which would require a
separate prescription for the matter sector.

The scalaron--Higgs example shows that a known two-field Einstein-frame sector can be represented by a regular generalized hybrid branch. In this case, the reconstructed hybrid function has a nonvanishing Hessian in the appropriate parameter domain, while the remaining inequalities select the region where the Einstein-frame conformal factor and the Palatini scalar are positive. The
example should therefore be understood as a structural embedding of the scalaron--Higgs scalar sector into GH gravity.

The inverse examples show how hybrid Einstein-frame potentials used in
cosmological applications are translated into classes of \(\Phi(R,\phi,X)\) representatives. In the power-law exponential case, the parameter controlling the hybrid fixed-point structure determines the powers of the curvature and scalar-field dependence in the reconstructed action. In the quadratic exponential case, the same procedure identifies the scalar-sector representatives associated with dark-sector hybrid potentials.

The resulting correspondence is therefore a local equivalence between regular vacuum Einstein-frame scalar sectors, rather than a global equivalence between arbitrary Jordan-frame actions. Within this domain, it provides a constructive map between \(\Phi(R,\phi,X)\) and \(f(R,\mathcal R)\) descriptions of the same two-scalar dynamics.

\printbibliography[heading=bibintoc]

@article{Ramirez2026,
title = {A connection between gravitational scalar-tensor theories and generalized hybrid theories},
journal = {Physics Letters B},
volume = {877},
pages = {140462},
year = {2026},
issn = {0370-2693},
doi = {https://doi.org/10.1016/j.physletb.2026.140462},
url = {https://www.sciencedirect.com/science/article/pii/S0370269326003151},
author = {Jonathan Ramírez and Santiago Esteban {Perez Bergliaffa}},
}

@article{PauloHybrid2020,
    author = "S{\'a}, Paulo M.",
    title = "{Unified description of dark energy and dark matter within the generalized hybrid metric-Palatini theory of gravity}",
    eprint = "2002.09446",
    archivePrefix = "arXiv",
    primaryClass = "gr-qc",
    doi = "10.3390/universe6060078",
    journal = "Universe",
    volume = "6",
    number = "6",
    pages = "78",
    year = "2020"
}

@article{Malik-2021,
title = {Existence of static wormhole solutions using f(R,$\phi$,X) theory of gravity},
journal = {New Astronomy},
volume = {89},
pages = {101632},
year = {2021},
issn = {1384-1076},
doi = {https://doi.org/10.1016/j.newast.2021.101632},
url = {https://www.sciencedirect.com/science/article/pii/S1384107621000646},
author = {Adnan Malik and Ayesha Nafees}
}

@article{Nojiri2010,
    author = "Nojiri, Shin'ichi and Odintsov, Sergei D.",
    title = "{Unified cosmic history in modified gravity: from F(R) theory to Lorentz non-invariant models}",
    eprint = "1011.0544",
    archivePrefix = "arXiv",
    primaryClass = "gr-qc",
    doi = "10.1016/j.physrep.2011.04.001",
    journal = "Phys. Rept.",
    volume = "505",
    pages = "59--144",
    year = "2011"
}

@article{Bahamonde2018,
    author = "Bahamonde, Sebastian and Bamba, Kazuharu and Camci, Ugur",
    title = "{New Exact Spherically Symmetric Solutions in $f(R,\phi,X)$ gravity by Noether's symmetry approach}",
    eprint = "1808.04328",
    archivePrefix = "arXiv",
    primaryClass = "gr-qc",
    doi = "10.1088/1475-7516/2019/02/016",
    journal = "JCAP",
    volume = "02",
    pages = "016",
    year = "2019"
}

@article{Garriga:1999vw,
    author = "Garriga, Jaume and Mukhanov, Viatcheslav F.",
    title = "{Perturbations in k-inflation}",
    eprint = "hep-th/9904176",
    archivePrefix = "arXiv",
    reportNumber = "UAB-FT-466",
    doi = "10.1016/S0370-2693(99)00602-4",
    journal = "Phys. Lett. B",
    volume = "458",
    pages = "219--225",
    year = "1999"
}

@article{Shinji2003,
    author = "Mukohyama, Shinji and Randall, Lisa",
    title = "{A Dynamical approach to the cosmological constant}",
    eprint = "hep-th/0306108",
    archivePrefix = "arXiv",
    reportNumber = "HUTP-03-A041",
    doi = "10.1103/PhysRevLett.92.211302",
    journal = "Phys. Rev. Lett.",
    volume = "92",
    pages = "211302",
    year = "2004"
}

@article{Sotiriou:2011dz,
    author = "Sotiriou, Thomas P. and Faraoni, Valerio",
    title = "{Black holes in scalar-tensor gravity}",
    eprint = "1109.6324",
    archivePrefix = "arXiv",
    primaryClass = "gr-qc",
    doi = "10.1103/PhysRevLett.108.081103",
    journal = "Phys. Rev. Lett.",
    volume = "108",
    pages = "081103",
    year = "2012"
}

@article{Amendola1993,
    author = "Amendola, Luca",
    title = "{Cosmology with nonminimal derivative couplings}",
    eprint = "gr-qc/9302010",
    archivePrefix = "arXiv",
    reportNumber = "PRINT-93-0202 (ROME)",
    doi = "10.1016/0370-2693(93)90685-B",
    journal = "Phys. Lett. B",
    volume = "301",
    pages = "175--182",
    year = "1993"
}

@article{Nilok2009,
  title = {Unified model of $k$-inflation, dark matter, and dark energy},
  author = {Bose, Nilok and Majumdar, A. S.},
  journal = {Phys. Rev. D},
  volume = {80},
  issue = {10},
  pages = {103508},
  numpages = {7},
  year = {2009},
  month = {11},
  publisher = {American Physical Society},
  doi = {10.1103/PhysRevD.80.103508},
  url ={https://link.aps.org/doi/10.1103/PhysRevD.80.103508}
}

@article{Nojiri2019,
    author = "Nojiri, S. and Odintsov, S. D. and Oikonomou, V. K.",
    title = "{$k$-essence $f(R)$ gravity inflation}",
    eprint = "1902.03669",
    archivePrefix = "arXiv",
    primaryClass = "gr-qc",
    doi = "10.1016/j.nuclphysb.2019.02.008",
    journal = "Nucl. Phys. B",
    volume = "941",
    pages = "11--27",
    year = "2019"
}

@article{Bahamonde2015,
    author = {Bahamonde, Sebastian and B{\"o}hmer, Christian G. and Lobo, Francisco S. N. and S{\'a}ez-G{\'o}mez, Diego},
    title = "{Generalized $f(R,\phi,X)$ Gravity and the Late-Time Cosmic Acceleration}",
    eprint = "1506.07728",
    archivePrefix = "arXiv",
    primaryClass = "gr-qc",
    doi = "10.3390/universe1020186",
    journal = "Universe",
    volume = "1",
    number = "2",
    pages = "186--198",
    year = "2015"
}

@article{Hwang2002,
    author = "Hwang, Jai-chan and Noh, Hyerim",
    title = "{Cosmological perturbations in a generalized gravity including tachyonic condensation}",
    eprint = "hep-th/0206100",
    archivePrefix = "arXiv",
    doi = "10.1103/PhysRevD.66.084009",
    journal = "Phys. Rev. D",
    volume = "66",
    pages = "084009",
    year = "2002"
}

@article{Sho2015,
    author = "Kaneda, Sho and Ketov, Sergei V.",
    title = "{Starobinsky-like two-field inflation}",
    eprint = "1510.03524",
    archivePrefix = "arXiv",
    primaryClass = "hep-th",
    reportNumber = "IPMU15-0177",
    doi = "10.1140/epjc/s10052-016-3888-0",
    journal = "Eur. Phys. J. C",
    volume = "76",
    number = "1",
    pages = "26",
    year = "2016"
}

@article{Malik2022,
  title = {A study of cylindrically symmetric solutions in $f(R,\phi,X)$ theory of gravity},
  author = {Malik, A. and Nafees, A. and Ali, A. and others},
  journal = {Eur. Phys. J. C},
  volume = {82},
  pages = {166},
  year = {2022},
  month = {02},
  publisher = {Springer},
  doi = {10.1140/epjc/s10052-022-10135-0},
  url = {https://doi.org/10.1140/epjc/s10052-022-10135-0}
}

@article{Anirudh2020,
    author = "Gundhi, Anirudh and Steinwachs, Christian F.",
    title = "{Scalaron-Higgs inflation}",
    eprint = "1810.10546",
    archivePrefix = "arXiv",
    primaryClass = "hep-th",
    reportNumber = "FR-PHENO-2018-013",
    doi = "10.1016/j.nuclphysb.2020.114989",
    journal = "Nucl. Phys. B",
    volume = "954",
    pages = "114989",
    year = "2020"
}

@article{Dhimiter2020,
  title = {A simple $F(R,\phi)$ deformation of Starobinsky inflationary model},
  author = {Canko, D. D. and Gialamas, I. D. and Kodaxis, G. P.},
  journal = {Eur. Phys. J. C},
  volume = {80},
  pages = {458},
  year = {2020},
  month = {05},
  publisher = {Springer},
  doi = {10.1140/epjc/s10052-020-8025-4},
  url = {https://doi.org/10.1140/epjc/s10052-020-8025-4}
}

@article{Maeda1989,
  title = {Towards the Einstein-Hilbert action via conformal transformation},
  author = {Maeda, Kei-ichi},
  journal = {Phys. Rev. D},
  volume = {39},
  issue = {10},
  pages = {3159--3162},
  numpages = {0},
  year = {1989},
  month = {05},
  publisher = {American Physical Society},
  doi = {10.1103/PhysRevD.39.3159},
  url = {https://link.aps.org/doi/10.1103/PhysRevD.39.3159}
}

@article{Rosa2017,
    author = "Rosa, Jo\~ao Lu\'\i{}s and Carloni, Sante and Lemos, Jos\'e Pizarro de Sande e and Lobo, Francisco Sab\'elio Nobrega",
    title = "{Cosmological solutions in generalized hybrid metric-Palatini gravity}",
    eprint = "1703.03335",
    archivePrefix = "arXiv",
    primaryClass = "gr-qc",
    doi = "10.1103/PhysRevD.95.124035",
    journal = "Phys. Rev. D",
    volume = "95",
    number = "12",
    pages = "124035",
    year = "2017"
}

@article{Sotiriou:2008rp,
    author = "Sotiriou, Thomas P. and Faraoni, Valerio",
    title = "{f(R) Theories Of Gravity}",
    eprint = "0805.1726",
    archivePrefix = "arXiv",
    primaryClass = "gr-qc",
    doi = "10.1103/RevModPhys.82.451",
    journal = "Rev. Mod. Phys.",
    volume = "82",
    pages = "451--497",
    year = "2010"
}

@ARTICLE{felice,
       author = {{De Felice}, Antonio and {Tsujikawa}, Shinji},
        title = "{$f(R)$ Theories}",
      journal = {Living Reviews in Relativity},
         year = 2010,
        month = jun,
       volume = {13},
       number = {1},
          eid = {3},
        pages = {3},
          doi = {10.12942/lrr-2010-3},
archivePrefix = {arXiv},
       eprint = {1002.4928},
 primaryClass = {gr-qc},
       adsurl = {https://ui.adsabs.harvard.edu/abs/2010LRR....13....3D}
}

@article{tamanini,
  title = {Generalized hybrid metric-Palatini gravity},
  author = {Tamanini, N. and B\"ohmer, C. G.},
  journal = {Phys. Rev. D},
  volume = {87},
  issue = {8},
  pages = {084031},
  numpages = {11},
  year = {2013},
  month = {},
  publisher = {American Physical Society},
  doi = {10.1103/PhysRevD.87.084031},
  url = {https://link.aps.org/doi/10.1103/PhysRevD.87.084031}
}

@article{Berti2015,
    author = "Berti, Emanuele and others",
    title = "{Testing General Relativity with Present and Future Astrophysical Observations}",
    eprint = "1501.07274",
    archivePrefix = "arXiv",
    primaryClass = "gr-qc",
    doi = "10.1088/0264-9381/32/24/243001",
    journal = "Class. Quant. Grav.",
    volume = "32",
    pages = "243001",
    year = "2015"
}

@article{Clifton2011,
    author = "Clifton, Timothy and Ferreira, Pedro G. and Padilla, Antonio and Skordis, Constantinos",
    title = "{Modified Gravity and Cosmology}",
    eprint = "1106.2476",
    archivePrefix = "arXiv",
    primaryClass = "astro-ph.CO",
    doi = "10.1016/j.physrep.2012.01.001",
    journal = "Phys. Rept.",
    volume = "513",
    pages = "1--189",
    year = "2012"
}

\section*{Acknowledgments}

This study was financed in part by the
Coordenação de Aperfeiçoamento de Pessoal de Nível Superior - Brasil (CAPES) - Finance Code 001.

\end{document}